\documentclass{article}
\usepackage{graphicx} 
\usepackage{hyperref}
\usepackage{mathtools}
\usepackage{tikz}
\usetikzlibrary{arrows.meta,
                decorations.markings}

\usepackage{geometry}
\mathchardef\mhyphen="2D 
\usepackage{appendix}
\usepackage{array}
\usepackage{amsmath,amssymb,amsthm}
\usepackage{mathrsfs}
\usepackage{tikz}
\usepackage{multicol}
\usepackage{algorithm}
\usepackage{bm}
\usepackage{algpseudocode}
\usepackage{amsmath}
 \usepackage{authblk}
\usepackage[style=alphabetic]{biblatex}
\newtheorem{theorem}{Theorem}
\newtheorem{corollary}{Corollary}
\newtheorem{lemma}[theorem]{Lemma}
\newtheorem{def.}{Theorem}

\newtheorem{condition}{Condition}
\newtheorem{taggedconditionx}{Condition}
\newenvironment{taggedcondition}[1]
 {\renewcommand\thetaggedconditionx{#1}\taggedconditionx}
 {\endtaggedconditionx}

\newtheorem{definition}[def.]{Definition}

\renewcommand{\Pr}{\mathbb{P}}
\newcommand{\E}{\mathbb{E}}
\newcommand{\LE}[2]{\textup{LE}(#1\rightarrow #2)}
\newcommand{\elfs}[2]{\textup{elfs}(#1,#2)}
\newcommand{\T}{\mathcal{T}}
\renewcommand{\Im}{\textup{Im}}

\usepackage{braket}
\title{Advances in Quantum Algorithms for the Shortest Path Problem}
\date{}
\author{Adam Wesołowski\footnote{\texttt{adam.wesolowski.2023@live.rhul.ac.uk}} }
\author{Stephen Piddock\footnote{\texttt{stephen.piddock@rhul.ac.uk}}}
\affil{Royal Holloway University of London, Department of Computer Science, UK}
\begin{document}

\maketitle

\begin{abstract}
  Given an undirected, weighted graph, with $n$ vertices and $m$ edges, and two special vertices $s$ and $t$, the problem is to find the shortest path between them. We give two bounded-error quantum algorithms with improved runtime in the adjacency list model that solve the problem on special classes of graphs defined via pathfinding probabilities of classical random walks and the electrical network framework. 
  Firstly, we give a simple quantum algorithm based on sampling edges from a graph via the quantum flow state and running a classical algorithm on the sampled edges. It runs in $\tilde{O}(l^2\sqrt{m})$ expected time and uses $O(\log{n})$ space on graphs where the shortest $s$-$t$ path is also a minimum resistance $s$-$t$ subgraph. 
  
  Our main algorithm can be thought of as a divide and conquer version of this approach and works on a special class of graphs where classical loop-erased random walk has a probability $q>0.537$ of finding the shortest $s$-$t$ path. In such cases the quantum algorithm outputs the shortest $s$-$t$ path with high probability in $\widetilde{O}(\ell\sqrt{m})$ expected time and $O(\log{n})$ space, where $l$ is the length (or total weight, in case of weighted graphs) of the shortest $s$-$t$ path.
This algorithm can be parallelised to  $\tilde{O}(\sqrt{lm})$ circuit depth when using $O(l\log{n})$ space. With the latter we partially resolve with an affirmative answer the open problem of whether a path between two vertices can be found in the number of steps required to detect it. 
\end{abstract}
\newpage

\section{Introduction}

Quantum algorithms have led to substantial improvements in a wide range of computational problems, including optimization, graph theory, and algebraic computation. For instance, polynomial speedups have been obtained for the Travelling Salesman Problem~\cite{moylett2017quantum}, min-cut~\cite{apersmincut}, connectivity~\cite{durr2006quantum}, graph coloring~\cite{gaspers2023quantum}, matching~\cite{ambainis2006quantum}, and path detection~\cite{belovs2012span, belovs2013quantum}. A more recent line of work has explored quantum approaches to path-finding, such as the algorithm of Jeffery et al.~\cite{path-edge}, and in this work, we contribute to this direction by studying quantum algorithms for the single-pair shortest path (SPSP) problem under specific graph-theoretic conditions.

\vspace{0.5em}
The classical shortest path problem is a foundational topic in algorithms and graph theory. It appears in several formulations, such as single-source shortest paths (SSSP), all-pairs shortest paths (APSP), and single-pair shortest path (SPSP). While Dijkstra’s algorithm~\cite{dijkstra}, and later improvements using Fibonacci heaps~\cite{fibbonaciheap} and Thorup’s linear-time algorithm for integer-weighted graphs~\cite{single-source}, have led to nearly optimal classical solutions, the quantum complexity of shortest path problems remains less well understood.

In this work we focus on the single-pair-shortest-path (SPSP) problem, formulated as follows:

\begin{definition}[Undirected Single-Pair Shortest Path (SPSP)]
Given an undirected graph $G = (V, E, w)$ with non-negative edge weights $w : E \rightarrow \mathbb{R}_{\geq 0}$ and a pair of vertices $s, t \in V$, the goal is to compute a path from $s$ to $t$ of minimum total weight (if one exists). That is, find a sequence of vertices $(v_0 = s, v_1, \dots, v_k = t)$ such that $\sum_{i=0}^{k-1} w(v_i, v_{i+1})$ is minimized over all such paths.
\end{definition}
While the SPSP problem is natural and frequently encountered in practice, it has not received significant dedicated attention in the classical literature. This is because, classically, SPSP requires essentially the same asymptotic runtime as SSSP, which in the worst case requires examining $\Theta(m)$ edges. As a result, SSSP algorithms are typically used to solve SPSP as a special case. However, in the quantum setting, SPSP may admit lower complexity than SSSP. This is to the best of our knowledge first work to focus on quantum algorithms for the SPSP problem.

\vspace{0.5em}
For the SSSP problem, D\"{u}rr et al.~\cite{durr2006quantum} designed a quantum algorithm that applies Grover's minimum finding~\cite{durr1996quantum} within Dijkstra's classical framework, achieving runtime $O(\sqrt{nm} \log^{3/2} n)$. Other quantum approaches have targeted restricted cases like acyclic graphs~\cite{Khadiev2019QuantumAF}, or employed quantum annealing heuristics~\cite{9186612} which do not yield asymptotic improvements.

\vspace{0.5em}
A closely related problem is $s$-$t$ connectivity: determining whether a path exists between two nodes. In the adjacency list model, D\"{u}rr et al.~\cite{durr2006quantum} provided an algorithm that not only detects connectivity in $\Theta(n)$ queries to the adjacency list (and $\Theta(n^{1.5})$ queries to the adjacency matrix), but also outputs an $s$–$t$ path. However, subsequent results~\cite{belovs2012span, belovs2013quantum} achieved faster query complexity for path detection (without providing a path). Belovs and Reichardt~\cite{belovs2012span} gave an algorithm in the adjacency matrix model that detects the presence of an $s$-$t$ path in $O(n\sqrt{L})$ queries, where $L$ is the length of the path. Belovs~\cite{belovs2013quantum} extended this to the adjacency list model, showing $O(\sqrt{RW})$ query complexity for detecting $s$-$t$ connectivity, where $R$ is the effective resistance and $W$ the total weight. In unweighted graphs, this yields an upper bound of $O(\sqrt{lm})$ for path detection, where $l$ is the length of the shortest path.

\vspace{0.5em}
    A natural question that emerged from these works is whether the quantum query complexity of \emph{finding} a path could match that of \emph{detecting} it. This question has been open for a long time with Jeffery et al.~\cite{path-edge} recently achieving an $\tilde{O}(L^{3/2}n)$ upper bound for single pair path-finding in the adjacency matrix model, where $L$ is the length of the longest path connecting $s$ and $t$. However, no known algorithm in the adjacency list model achieves query complexity matching Belovs’ detection bound.

\vspace{0.5em}
 We give a quantum algorithm for finding shortest paths in the adjacency list model with a quantum circuit depth which asymptotically matches Belovs' detection algorithm~\cite{belovs2013quantum}, \emph{if one of two specific structural conditions on the input graph holds}. Informally, the first condition (Condition~\ref{cond:1}) requires that the shortest $s$-$t$ path is also a subgraph with minimum $s$-$t$ resistance\footnote{such that any subgraph that contains the entire shortest $st$ path has lower resistance than any subgraph that does not contain the entire  path}. 
 Condition \ref{cond:2} requires that a classical \textit{loop-erased random walk} can find the shortest $s$-$t$ path with probability $p>0.537$.
These two conditions do not appear to be directly comparable, but Condition 1 appears in some ways stronger. 
 If one of these two conditions hold, then our quantum algorithm can find that path in $\widetilde{O}(\ell\sqrt{m})$ expected depth with high probability.   Crucially, we emphasize that our algorithm \textbf{does not apply to general graphs}, and care must be taken in interpreting our results within the class of graphs covered by conditions~\ref{cond:1} or~\ref{cond:2}. 
 A general, unconditional lower bound of $\Omega(n)$ may be obtained for the SPSP problem via a trivial reduction from the undirected st-connectivity problem~\cite{apers2022no}.

Under the imposed conditions, we construct two quantum algorithms for the SPSP problem. The first, $\mathcal{A}_1$, is a brute-force approach with query complexity $\tilde{O}(l^2\sqrt{m})$. The second, $\mathcal{A}_2$, uses a divide-and-conquer strategy to reduce the cost to $\tilde{O}(l\sqrt{m})$ steps and $O(\log{n})$ memory. We mention that with $O(l\log{n})$ work space we can parallelize the divide and conquer approach to obtain an algorithm with circuit depth $\tilde{O}(\sqrt{lm})$), where $l$ is the length of the path. This provides evidence that matching path detection complexity for shortest path-finding can be achieved in special class of graphs. For comparison, Jeffery et al.'s algorithm~\cite{path-edge} in the adjacency matrix model requires $\tilde{O}(L^{3/2}n)$ queries; for small $l$, our result is strictly better.

\vspace{0.5em}
The SPSP problem is fundamental in a range of applications such as navigation~\cite{madkour2017survey}, autonomous driving~\cite{jiao2019using}, network routing, and bioinformatics~\cite{1565664}. Moreover, shortest paths are often used as subroutines in more complex algorithms~\cite{Delling2009,6194335,7171911,1565664,hermansson2015generalized}. Our results imply a speedup for such subroutines under structural assumptions. Notably, they improve the performance of the shortest-path kernel method by Borgwardt and Kriegel~\cite{1565664}.

\paragraph{Results.}
The main contribution of this work is a demonstration of quantum algorithms that on structured instances find the shortest $s\mhyphen t$ path in query complexity lower than any known classical and quantum algorithms.
Algorithms $\mathcal{A}_1, \mathcal{A}_2$ represent respectively, a brute force approach and a divide and conquer approach. Importantly, we note that if the length of the shortest path scales at most \textit{polylogarithmically} with the instance size then algorithm $\mathcal{A}_2$ with efficient $\log{n}$ space finds the shortest $s\mhyphen t$ path in query complexity asymptotically equal to the complexity of Belovs~\cite{belovs2013quantum} algorithm for detecting a path. Thereby, we partially resolve the longstanding open problem of whether there exist path-finding quantum algorithms with query complexity equal (up to \textit{polylog} factors) to the cost of quantum path detection algorithms.
Nevertheless, both our algorithms require structured instances. For the description of the class of graphs, including necessary conditions see section~\ref{sec:conditions}. 
In the paper, in sections~\ref{sec:algorithm1} and~\ref{sec:algorithm2}, respectively, we prove the following theorems:

\begin{theorem}[Simple algorithm for the shortest path finding]
    \label{thm:simplealg}
   Given quantum access to the adjacency list of a graph $G$ and two vertices $s$ and $t$ satisfying Condition~\ref{cond:1}, there exists a quantum algorithm $\mathcal{A}_1$ that with high probability finds the shortest $s\mhyphen t$ in  $\tilde{O}(l^{2}\sqrt{m})$ steps using $O(\log{n})$ qubits and $O(l\log{l})$ classical bits.
\end{theorem}
\begin{theorem}[Divide and conquer quantum algorithm for the shortest path finding]
\label{thm:divideandconquer}
   Given quantum access to the adjacency list of a graph $G$ and two vertices $s$ and $t$ satisfying Condition~\ref{cond:2} or Condition~\ref{cond:1'}, there exists a bounded-error quantum algorithm $\mathcal{A}_2$ that with high probability outputs the shortest $s$-$t$ path in $\tilde{O}(l\sqrt{m})$ steps using $O(\log{n})$ qubits. 
   This algorithm can be parallelised to run in $\tilde{O}(\sqrt{lm})$ depth using $O(l\log{n})$ qubits.  
\end{theorem}

The resources used by our algorithms are concisely summarized in the table below. 

\begin{table}[h!]
    \centering
    \begin{tabular}{|c c c|}
    \hline
     Algorithm's base primitive&  space complexity & circuit depth\\
     \hline
     $\mathcal{A}_1$: sampling the flow state & $O(\log{n})$ & $\tilde{O}(l^2\sqrt{m})$ \\
    $\mathcal{A}_2$: divide and conquer &  $O(\log{n})$ &  $\tilde{O}(l\sqrt{m})$\\
    $\mathcal{A}_2$: divide and conquer (parallelized)&   $O(l\log{n})$ &  $\tilde{O}(\sqrt{lm})$\\
 
    \hline
    \end{tabular}
\end{table}
In table 1. we provide a comparison of our results to the best classical algorithm and best known quantum algorithms. 
\begin{table}[h!]
    \centering
    \begin{tabular}{|c c c c c|}
        \hline
       & Classical & D\"{u}rr et al.~\cite{durr2006quantum} & Jeffery et al.~\cite{path-edge} & this work \\
        \hline
        $s$-$t$ path & $O(m)$ & $\Theta(n)$ & $\tilde{O}(L^{\frac{3}{2}}n)$ & $ \tilde{O}(l\sqrt{m})$ \& $\tilde{O}(\sqrt{lm})$ \\
        
        shortest $s$-$t$ path & $O(m)$ &  $\tilde{O}(\sqrt{nm})$     &--- & $ \tilde{O}(l\sqrt{m}) $ \& $\tilde{O}(\sqrt{lm})$ \\
    
        \hline
    \end{tabular}
    \caption{A table with comparison of algorithms $\mathcal{A}_1$ and $\mathcal{A}_2$ presented in this work, followed by the comparison to other algorithms. In the most bottom table the comparison is given in terms of query complexity (in the adjacency list model except for the result by Jeffery et al. which is in the adjacency matrix model.). Lower-case $m$, $n$ and $l$ stand for the number of edges, vertices and length of the shortest $s\mhyphen t$ path, respectively. Upper-case $L$ is the length of the longest $s$-$t$ path in $G$. }
    \label{tab:my_table}
\end{table}\\

Our results can be compared to the adjacency matrix quantum algorithm of Jeffery et al.~\cite{path-edge} that outputs a path in $\tilde{O}(L^{3/2}n)$ steps, where $L$ is the length of the longest path in $G$. The results of~\cite{path-edge} apply in the adjacency matrix model so cannot be strictly compared to ours, but one can do a rough comparison by replacing $n$ by $\sqrt{m}$ since this is the difference in cost of preparing the quantum flow state in the two models.
In our restricted-problem-instance framework for the $l$ factor we obtain an exponent of `2' in the scaling for algorithm $\mathcal{A}_1$ and `1' for algorithm $\mathcal{A}_2$, where the $l$ factor corresponds to the shortest path length. 

We also make a simple observation that the exponent can be decreased further to $\frac{1}{2}$ if we have access to linear quantum work space and parallelise the subproblems during the divide and conquer algorithm. 
We also note that using the framework assumed in this work there are no quantum algorithms with query complexity lower than $O(\sqrt{lm}(\frac{1}{\varepsilon}+\log{l}))$, which corresponds to the cost of single quantum flow state preparation. One can achieve this lower bound up to \textit{poly(log)} simply by parallelization of the procedures in algorithm $\mathcal{A}_2$. The procedure of parallelization can be applied to any algorithm, and for algorithm $\mathcal{A}_1$, we can parallelise to achieve a circuit depth of $\tilde{O}(l\sqrt{m})$, equal to the cost of one flow state preparation (in algorithm $\mathcal{A}_1$, this takes longer because the state needs to be prepared to greater accuracy).


\vspace{1em}
\textbf{Techniques.} 
Our techniques rely on primitives from quantum electrical resistance estimation~\cite{apers2022elfs}, which allow for approximate preparation of a quantum state representing the electric $s$–$t$ flow. We leverage this to design bounded-error quantum procedures that extract the shortest path, assuming one of the conditions hold.
Both Condition 1 and Condition 2 ensure that in the electric flow from $s$ to $t$, there is a large amount of flow along the shortest $s$-$t$ path, which means that if you sample the electric flow (by measuring the quantum flow state as described in~\cite{apers2022elfs}), then there is a good probability of sampling an edge on the shortest $s$-$t$ path.
See lemmas \ref{lem:flowhalf}, \ref{lem:RG>RP/2}, \ref{lem:probinsubgraph} and \ref{lem:probonpath} in Section~\ref{sec:conditions} for more precise statements of the consequences of the two conditions.

Algorithm $\mathcal{A}_1$ repeatedly samples the $s$-$t$ flow until enough samples have been taken to cover every edge on the shortest $s$-$t$ path, using a coupon collector argument to derive the required number of samples. Then a classical shortest path algorithm is run on the sampled edges. This is only proved to work under Condition 1, because this condition implies a stronger result on the concentration of electric flow along the shortest $s$-$t$ path, namely that in the unit electric $s$-$t$ flow, there is a flow of at least $1/2$ along every edge of the shortest $s$-$t$ path.

Algorithm $\mathcal{A}_2$ works with a divide-and-conquer strategy, first sampling a random edge $e$ from the electric $s$-$t$ flow and then using a random endpoint $x$ of $e$ to subdivide the problem into finding an $s$-$x$ shortest path and an $x$-$t$ shortest path which can be solved recursively. This only works if $x$ is on the shortest $s$-$t$ path, and we ensure this is (likely) the case by sampling multiple edges and choosing the one that maximises the $s$-$t$ effective resistance when removed from the graph. See sections \ref{sec:algorithm1} and \ref{sec:algorithm2} for a full description of the algorithms.

\vspace{1em}
\paragraph{Organization.} In Section~\ref{sec:prerequisites} we cover the necessary background material on electric networks, quantum flow states, loop-erased classical random walks and uniform spanning trees. Section~\ref{sec:conditions} defines the conditions under which our algorithms works. Sections \ref{sec:algorithm1} and \ref{sec:algorithm2} detail the two algorithms and their analyses. We conclude in Section~\ref{sec:conclusion} with discussion and future directions.

\section{Prerequisites}
\label{sec:prerequisites}
\subsection{Notation}
Throughout the work we make use of a lot of variables and notation conventions. For clarity we outline them in the section. $R_G(s,t)$- stands for the value of effective resistance between nodes $s$ and $t$ in graph $G$. $R_{G\setminus e}(s,t)$ stands for the value of effective resistance between nodes s and t in a graph $G$ with and edge/subgraph e removed from $G$. Upper-case $P$ refers to the shortest $s\mhyphen t$ path, i.e $R_P(s,t)$ is the effective resistance of the shortest path which corresponds to the length (or total weight for weighted graphs) of the shortest path which we denote by a lower-case $l$. Throughout the work we refer to $l$ as the length of the shortest path, but the extension to weighted graphs requires only renaming $l$ to minimum weight path, and does not affect the analysis, and every claim extends straightforwardly to weighted graphs.
Lower case $m$ and $n$, refer to, respectively, the number of edges and vertices in a graph $G$. Upper case $L$ denotes the longest $s\mhyphen t$ path in a graph, $d_v$ degree of vertex $v$, $\Delta$ is a maximal degree of a vertex in $P$. All other notations, such as $\varepsilon$, $\delta$ and '$p$' are used in several places and defined locally in the paper. 

\subsection{Adjacency list vs. Adjacency matrix}

The \emph{adjacency list} representation of a graph \(G = (V, E)\) consists of an array of lists. For a graph with \(n\) vertices and \(m\) edges, this structure includes \(n\) lists, one for each vertex, where each list contains all adjacent vertices. This model has a space complexity of \(O(n + m)\), making it advantageous for sparse graphs where the number of edges \(m\) is much less than \(n^2\). Operations such as edge insertion and deletion can be performed in constant time, while checking for the existence of an edge and iterating over all neighbors of a vertex take time proportional to the degree of the vertex.

Conversely, the \emph{adjacency matrix} representation uses a 2D array of size $n \times n$, where each element indicates the presence or absence of an edge between vertex pairs. This model is suitable for dense graphs where the number of edges is close to $n(n-1)/2$, as it requires $O(n^2)$ space. It allows for constant-time complexity for edge existence checks, and both edge insertion and deletion. However, iterating over the neighbors of a vertex requires $O(n)$ time classically.

In a quantum setting, in an adjacency list model one assumes access to an oracle which given a vertex $u$ and an index $i$ outputs the $i^{th}$ vertex of $u$; and in adjacency matrix model the oracle given a query consisting of a pair of vertices tells whether there is an edge between them.
Quantum adjacency list access to a graph allows for easy implementation of the quantum walk operator, which allows us to directly use the results of~\cite{piddock2019quantum,apers2022elfs} which count the cost of their algorithms in terms of the number of quantum walk steps.

\subsection{Electric network formalism}
In this short section we outline several standard facts about the electrical network framework for graphs. This subsection is based on theory outlined in a more comprehensive overview of the topic by Bollobas~\cite{Bollobas1998Modern}.
We denote weighted graphs $G=(V,E,w)$ with vertex set $V$, edge set $E$, and non-negative edge weights $w: E\rightarrow \mathbb{R}^{+}$.
It is convenient to consider such a graph as an electrical network where each edge $e$ represents a resistor of resistance $r_e = 1/w(e)$.

For each edge in the graph we apply an arbitrary direction to get a directed edge set $\vec E$. A flow on $G$ is a vector $f \in \mathbb{R}^{\vec E}$.
Across all possible flows through a network $G$, we identify the electric flow as the one that minimises the energy dissipation function. 

\begin{definition}
    The energy of a flow $f$ in a network $G$ is defined to be:
    \begin{equation}
        \mathcal{E}(f)=\sum_{e\in E} f_{e}^2r_{e}
        \label{eq:energy}
    \end{equation}
     
\end{definition}

By considering $f$ as a vector on the set of directed edges $\vec{E}$, we can write the energy as \begin{equation}
    \mathcal{E}(f)= f \cdot f.
    \label{eq:energydotproduct}
\end{equation}
Here the inner product $f\cdot g$ denotes $f^{\top} W^{-1} g$ where $W^{-1}$ is a diagonal matrix with entries $r_e$ on the diagonal (for an unweighted graph $W=I$ and this is the usual dot product $f\cdot g=f^{\top} g$).

\begin{definition}
The electric flow is precisely the flow that minimises the function 
    $\mathcal{E}(f)$. The electric flow obeys Kirchhoff's laws and Ohm's law.
    In a graph $G$, the effective resistance between $s$ and $t$ is the energy of the unit electric flow from $s$ to $t$:
    \begin{equation}
        R_{G}(s,t)=\sum_e f_{e}^2r_e=\mathcal{E}(f)
    \end{equation}
\end{definition}

If the potential difference $\Delta\phi_e$ is present across an edge $e$ of resistance $r(e)$, the electric current of magnitude $f_e=\frac{ \Delta\phi_e}{r(e)}$ will flow through edge $e$, in accordance with \textit{Ohm's law}. \textit{Kirchhoff's potential law} asserts that across any cycle $C$ in $G$ the sum of potential drops will be equal to zero: $\sum_{e\in C}\Delta\phi_e=0$. \textit{Kirchhoff's current law} states that the total current is conserved across the network. The electrical flow can be easily shown to satisfy both Kirchoff's and Ohm's laws.

From \textit{Rayleigh’s monotonicty principle} it follows that removing an edge of resistance $r_e$ from the network $G$ never decreases effective resistance between any pair of nodes. It is also well known that similarly to the standard edge-distance (also refereed to as topological distance), the resistance distance also constitutes a metric on $G$. Nevertheless, resistance distance does not correspond uniquely to the topological distance, i.e there are plenty of possible graph structures that have different topological distance between two nodes but the same resistance distance.

For paths connected in parallel the effective $s\mhyphen t$  resistance is always smaller then the topological length of the shortest path connecting $s$ and $t$. It is also true that if there is a one short path of length $l$, then connecting another long path of length $L$ in parallel does not alter the $s\mhyphen t$ effective resistance substantially as long as $L>>l$. This will be of importance when we introduce $condition \hspace{0.1cm}1$, which puts a lower bound on how small can the effective $s\mhyphen t$ resistance be in the graph $G$. We refrain from providing a more comprehensive overview of the electrical network theory, in order to retain the relative compactness of the entire work. Readers wishing to investigate the theory further we refer to the standard textbook of Bollobas~\cite{Bollobas1998Modern}.

\subsection{The quantum flow state and estimation of the effective resistance}
The first connection between quantum walks on graphs and electric flow in a network has been established by Belovs~\cite{belovs2013quantum}. Later on, Piddock~\cite{piddock2019quantum,apers2022elfs} showed that one can efficiently obtain a quantum electric flow state for an arbitrary finite graph. The flow state is a weighted superposition of edges on a graph $G$ with weights proportional to the amount of electric flow going through a given edge. The state is normalised by the square root of the effective resistance between the source vertex $s$ and the sink vertex $t$.
\begin{equation}
    \ket{f_G^{s\rightarrow t}}=\frac{1}{\sqrt{R_{G}(s,t)}}\sum_{e \in \vec{E}} f_e\sqrt{r_e}\ket{e}
    \label{eq:electricflowstate}
\end{equation}
The procedure of preparing $\ket{f^{s\rightarrow t}}$ is based on running a phase estimation on a quantum walk operator and an antisymmetrized star state of the starting vertex $s$. The star state for a vertex $x$ is defined as $\ket{\phi_x}=\frac{1}{\sqrt{d_x}}\sum_y\sqrt{w_{xy}}\ket{xy}$ where $y$ are vertices connected to $x$ by and edge with weight $w_{xy}$. Then, given a graph with two special vertices called a source $s$ and a sink $t$ one defines two unitary reflection operators: 
\begin{equation}
    D_x= 2\ket{\phi_x}\bra{\phi_x}-I
\end{equation}
and 
\begin{equation}
    \bigoplus_{x\notin \{s,t\} } D_x
\end{equation}
The first one represents reflections around the star state of vertex $x$ and the second one represents a reflection about the span of all star states but $\ket{\phi_s}$ and $\ket{\phi_t}$.
Together with the SWAP operator $SWAP\ket{xy}=\ket{yx}$, the operators form an absorbing quantum walk operator $U$: \begin{equation}
    U=SWAP\cdot \bigoplus_{x\notin \{s,t\} } D_x
\end{equation}
In their work~\cite{piddock2019quantum,apers2022elfs} Piddock observed that the projection of an antisymmetrised star-state of the source vertex $s$ $\ket{\phi^-_s}=\frac{I-SWAP}{\sqrt{2}}\ket{\phi_s}$ onto the invariant subspace of operator $U$, yields the electric flow state $\ket{f^{s\rightarrow t}}$. Let $\Pi_U$ be a projection on the invariant subspace of the unitary $U$, then \begin{equation}
    \Pi_U \ket{\phi_s^-}= \frac{1}{\sqrt{R_G(s,t)d_s}}\ket{f^{s\rightarrow t}}
\end{equation}

It is shown in~\cite{piddock2019quantum,apers2022elfs} that performing phase estimation on $\ket{\phi_s^-}$ with operator $U$ to precision $\delta= \frac{\varepsilon}{R(s,t)m}$, and measuring the phase register to be $'0'$, leaves the state final state $\ket{f}$ $\varepsilon$- close to the flow state $\ket{f^{s\rightarrow t}}$ with probability $p \geq \frac{1}{R(s,t)d_s}$. The expected number of queries to obtain an $\varepsilon$-approximation of $\ket{f^{s\rightarrow t}}$ is shown to be equal to $O(\frac{\sqrt{R(s,t)m}}{\varepsilon})$.
From the Theorem 8 from~\cite{apers2022elfs} we know that:
\begin{theorem}[Apers and Piddock~\cite{apers2022elfs}]
\label{thm:prepareflowstate} 
    Given an upper bound on the escape time $ET_{s}$ from vertex $s$ $\hat{ET}_{s} \geq ET_{s}$, there is a quantum walk algorithm that returns the state $\ket{f^{'}}$ such that $||\ket{f^{'}}-\ket{f}||_2\leq \varepsilon_1$. It requires $O(\sqrt{ET_s}(\frac{1}{\varepsilon_1}+ \log{(R_G(s,t)d_s})) )$ queries to the quantum walk operator. 
    
\end{theorem}
The escape time $ET_s$ is always upper bounded by \begin{equation}
    ET_s< R_G(s,t)m 
\end{equation} because \begin{equation}
    ET(s)\leq HT(s,t)< CT(s,t) =2R_G(s,t)m 
\end{equation}
Where $HT(s, t)$ is the hitting time from $s$ to $t$, and $CT(s, t)$ is the commute time from $s$ to $t$. The last equality was shown to be true by~\cite{commuteandcover}.  It is shown that after the expected number of calls to the quantum walk operator $U$ phase estimation procedure succeeds (outputs an $\varepsilon_1\mhyphen approximation$ in the Hilbert-Schmidt norm of the flow state, i.e 
 \begin{equation}
     || \ket{f^{'}}-\ket{f} ||_2 \leq \varepsilon_1
 \end{equation} where $\ket{f^{'}}$ is the prepared state).

In~\cite{apers2022elfs} the authors outline a procedure that allows to $\varepsilon\mhyphen multiplicatively$ estimate the effective resistance between two vertices in a graph $G$. The following theorem encapsulates that fact:\\
\begin{theorem}[Apers and Piddock~\cite{apers2022elfs}]
\label{thm:estimateresistance}    Given an upper bound on the escape time from vertex s $\hat{ET} \geq ET$, there is a quantum walk algorithm that
$\varepsilon\mhyphen multiplicatively$ estimates $R_G(s,t)d_s$, where $d_s$ is a degree of vertex $s$, using a number of quantum walk steps
$$O(\sqrt{\hat{ET}}(\frac{1}{\varepsilon^{\frac{3}{2}}}+\log{(R_G(s,t)d_s)}))$$
\end{theorem}
The algorithm is based on estimation of an amplitude of the flow state in the output state of the quantum phase estimation procedure. The amplitude is inversely proportional to the effective resistance $R_G(s,t)$ such that applying quantum amplitude estimation returns $\varepsilon\mhyphen multiplicative$ approximation of $R_G(s,t)$.

\subsection{Uniform Spanning Trees}

Consider a (classical) random walk algorithm that finds the shortest path between $s$ and $t$ by running the \emph{loop erased} random walk from $s$ to $t$. 
Let $\Pr_{LE(s\rightarrow t)}(P)$ be the probability that the loop erased walk from $s$ to $t$ returns the path $P$. 

Let $\T(G)$ be the set of spanning trees of $G$. For an unweighted graph, we are interested in the random process of uniformly sampling a tree from $\T(G)$.

\subsubsection{Uniform spanning trees on weighted graphs}
This is generalised to weighted trees as follows. For a tree $T$ of $G$ we define the weight of $T$ to be the product of the weights of all the edges in $T$ and write:
$w(T)=\prod_{e\in T} w(e)$.
For a set of trees $\T$, define $w(\T)=\sum_{T\in \T}w(T)$ be the sum of the weights of all the trees in $\T$. Now consider the distribution $\mu$ defined by taking a tree $T$ with probability 
\[\Pr_{T'\sim \mu}[T'=T]=\frac{w(T)}{w(\T(G))}.\]
Note that for an unweighted graph, $w(T)=1$ for any tree $T$ and $w(\T)=|\T|$ for any set of trees $\T$, so this distribution matches the uniform distribution on $\T(G)$.

\subsubsection{Loop erased random walk and uniform spanning trees}
There is a tight connection between the loop erased random walk and random sampling of spanning trees. For a path $P$ between vertices $s$ and $t$, the probability that the loop erased random walk from $s$ to $t$ takes the path $P$ is equal to the probability that a random spanning tree contains the path $P$:
\begin{equation}
    \Pr_{P'\sim \LE{s}{t} }[P'=P] = \Pr_{T\sim \mu}[P\subseteq T]=\frac{w(\{T \ | \ T \in \T(G) \text{ and } T\supseteq P\}) } {w(\T(G))}=\frac{w(\T_P)}{w(\T(G))}
    \label{eqn:LEprobtrees}
\end{equation}
where we have introduced the notation $\T_P=\{T \ | \ T \in \T(G) \text{ and } T\supseteq P\}$ for the set of spanning trees that contain the path $P$.

\subsubsection{Effective resistance and spanning trees}
There is also a close connection between spanning trees and effective resistance. 
\begin{equation}
    R_{G(s,t)}=\frac{w(\T(G/\{s,t\}))}{w(\T(G))}
    \label{eqn:Rspanningtrees}
\end{equation}
where $G/\{s,t\}$ denotes the graph $G$ with vertices $s$ and $t$ contracted to a single vertex. 
Observe that the trees in $\T(G/\{s,t\})$ can be interpreted as spanning forests of $G$, with two connected components separating $s$ and $t$.

If $s$ and $t$ are the two endpoints of an edge $e$, then any tree $T \in \T(G/\{s,t\})$ can be identified with a spanning tree $T\cup \{e\}$ of $G$, by adding the edge $e$.
Since $w(T\cup e)=w(e)w(T)$, we have $w(\T(G/\{s,t\}))=w(\T_{\{e\}})/w(e)$.
Therefore, when $s$ and $t$ are endpoints of the edge $e$:
\begin{equation}
    R_{G(s,t)}=\frac{w(\T(G/\{s,t\}))}{w(\T(G))}=r_e\frac{w(\T_{\{e\}})}{w(\T(G))}=r_e\Pr_{T\sim \mu}[e \in T]
    \label{eqn:PeinT}
\end{equation}

\subsubsection{Negative association}
One important property of the uniform spanning tree distribution is \emph{negative association} (see Theorem 2.1 of~\cite{Grimmett_2018}). If $A,B$ are disjoint subsets of the edge set $E$, then, if $T$ is a random spanning tree:
\begin{equation}
    \Pr[(A\cup B) \subseteq T] \le \Pr[A\subseteq T]\Pr[B\subseteq T]
    \label{eqn:negativeassociation}
\end{equation}
So for a path $P$, 
\[\Pr[P\subseteq T]\le \prod_{e\in P}\Pr[e \in T]\]

\section{Class of graphs}
\label{sec:conditions}
We consider two conditions which characterise the types of graphs where our path finding algorithms work efficiently.

\subsection{Shortest path by effective resistance distance}
Intuitively, our first condition requires that there exists a path $P$ between $s$ and $t$ which is shorter than all other routes between $s$ and $t$, even when measured using the effective resistance distance. 
When computing the effective resistance distance of the alternative routes between $s$ and $t$, we wish to consider multiple paths at once and so we compare $P$ against other subgraphs $X \subset G$.


Let $P$ be the shortest path between vertices $s$ and $t$.
\begin{condition}
    \label{cond:1}
    For any subgraph $X \subset G$ such that $P \not\subseteq X$, we have $R_X(s,t) > R_P(s,t)$.
\end{condition}

This condition requires that $P$ be the unique shortest path between $s$ and $t$ when measured using the resistance distance, compared not only to all other paths but all subgraphs that do not contain the whole $P$.
It follows from Condition~\ref{cond:1} that:
\begin{corollary}
    \label{cor:cond1}
     For any edge $e \in P$,  we have $R_{G\setminus e}(s,t) \geq R_P(s,t)$.
\end{corollary}
In fact it is not hard to show that Corollary~\ref{cor:cond1} is equivalent to condition~\ref{cond:1}, by noting that for $e \in P\setminus X$, we have $R_X \geq R_{G\setminus e}$ by Rayleigh's monotonicity principle.

Condition~\ref{cond:1} turns out to be a very strong condition. In the appendix we prove that 

\begin{lemma}
    \label{lem:flowhalf}
    Condition~\ref{cond:1} implies that the electric flow through every edge $e\in P$ satisfies:
    \begin{equation}
        f_{e}\geq 1/2
    \end{equation}
\end{lemma}

\begin{lemma}
    \label{lem:RG>RP/2}
    Let $P$ be a path in $G$ such that Condition~\ref{cond:1} holds. Then $R_G(s,t) \geq R_P(s,t) / 2$.
\end{lemma}

This is useful because of the following lemma:
\begin{lemma}
    \label{lem:probinsubgraph}
    Let $G$ be a graph and $X$ a subgraph of $G$ such that $s$ and $t$ are connected in $X$. Then if $e$ is a random edge sampled from the electric flow state:
    \[\Pr_{e \sim \elfs{s}{t}}[e\in X]\geq\frac{R_{G}(s,t)}{R_{X}(s,t)},\]
    where $R_X(s,t)$ is the effective resistance between $s$ and $t$ in the subgraph $X$.
    
    \end{lemma}
    \begin{proof}
         The probability of sampling an edge $e$ from the electric flow state is  is $\frac{f_e^2r_e}{R_G(s,t)}$, which follows directly from the structure of the flow state (eq.\eqref{eq:electricflowstate}). 
         For the remainder of this proof, we drop the $(s,t)$ labels for convenience. Let the voltages of the electric flow $f$ be $\{\phi_x\}_{x \in G}$, with $\phi_s=R_G$ and $\phi_t=0$. 
         By Ohm's Law the flows are related to the voltages as $f_{xy}=(\phi_x - \phi_y)/r_{xy}$ and so the probability of sampling an edge in $X$ is:
        \begin{equation}
            \Pr[e\in X]= \frac{1}{R_G}\sum_{xy\in X} \frac{(\phi_x-\phi_y)^2}{r_{xy}}
            =\frac{R_G}{R^2_X}\cdot \sum_{xy\in X} \frac{1}{r_{xy}}(\phi_x\frac{R_X}{R_G}-\phi_y\frac{R_X}{R_G})^2
        \end{equation}
        The sum here is equal to the energy of the potentials $\{\phi'_x = \phi_x \frac{R_X}{R_G}\}_{x \in X}$ on the subgraph $X$, which take the values $\phi'_s=R_X$ at $s$ and $\phi'_t=0$ at $t$.
        
        These potentials define a flow on $G$, but do not define a flow on $X$. Nevertheless, by Theorem 1 of Chapter IX of~\cite{Bollobas1998Modern}, the energy of such potentials is minimised by the potentials of the electric flow, and thus is lower bounded by the effective resistance $R_X$.
        This gives
        \begin{equation}
            \sum_{xy\in X} \frac{1}{r_{xy}}(\phi_x\frac{R_X}{R_G}-\phi_y\frac{R_X}{R_G})^2\geq R_X
        \end{equation} 
       Thus we obtain the bound claimed in the statement of the lemma:
        \begin{equation}
            \Pr[e\in X]=\frac{R_G}{R^2_X}\cdot \sum_{xy\in X} \frac{1}{r_{xy}}(\phi_x\frac{R_X}{R_G}-\phi_y\frac{R_X}{R_G})^2\geq\frac{R_G}{R_X}
        \end{equation}
    \end{proof}
    
\subsection{Shortest path that can be found by classical random walk}
The second condition that we consider in this paper is designed to capture the case where a classical random walk would be likely to find the shortest path $P$.

\begin{condition}
\label{cond:2}
    The probability that a classical loop erased random walk returns the shortest path $P$ between $s$ and $t$ is at least $0.537$.
\end{condition}

\begin{lemma}
    \label{lem:probonpath}
    Let $P$ be an $st$ path, and let $e$ be sampled from the electric flow between $s$ and $t$.
    Then
\[\Pr_{e \sim \elfs{s}{t}}[e\in P] \ge \frac{R_G(s,t)}{R_P(s,t)} \ge \Pr_{P'\sim \LE{s}{t} }[P'=P]\]
\end{lemma}

\begin{proof}
    The first inequality is the special case of Lemma~\ref{lem:probinsubgraph} where the subgraph $X$ is the path $P$. For the second inequality, we use the spanning tree connection for the loop erased random walk path \eqref{eqn:LEprobtrees} and for the effective resistance \eqref{eqn:Rspanningtrees}.
    It therefore remains to show that:
    \[w(\T(G/\{s,t\})) \ge R_P \times w(\T_P) \]
    where $\T_P=\{T \ | \ T \in \T(G) \text{ and } T\supseteq P\}$.

    Let $T$ be a spanning tree of $G$ with $P\subseteq T$, and let $e \in P$. Observe that removing the edge $e$ from $T$ disconnects $s$ and $t$, and so $T\setminus e$ is an element of $\T(G/\{s,t\})$.
   The map $g$ \begin{align*}
         g: \T_P \times E &\to \T(G/\{s,t\}) \\
         (T,e) &\mapsto T\setminus e
    \end{align*} is an injection.

    Furthermore, the weight $w(T\setminus e)=w(T)/w(e)$ and so we have 
    \[w(\T(G/\{s,t\})) \ge w(\Im(g)) = \sum_{T\in \T_P} \sum_{e\in E} w(T\setminus e)=\sum_{e\in E} \frac{1}{w(e)} \sum_{T\in \T_P} w(T)=R_P w(\T_P)\]

\end{proof}

\section{Algorithm $\mathcal{A}_1$: simple approach}
\label{sec:algorithm1}
Under the condition 1 described above one may devise in a straightforward way an efficient algorithm, where a quantum procedure is followed by a classical (or quantum) post-processing. 
If one focuses only on repeated sampling of the quantum flow state $\ket{f_G^{s\rightarrow t}}$, then the assumed condition assures that $O(l\log{l})$ samples suffice to sample all edges from the shortest path with high probability. Then one can run a classical shortest path algorithm, such as Thorup, Dijkstra or Breadth-First-Search algorithms on the set of sampled edges or use algorithm $\mathcal{A}_2$ of this work. In this approach, the sampling procedure serves as a kind of graph sparsifier that with high probability preserves the shortest $s\mhyphen t$ path. Thereby, the quantum procedure reduces the instance size from $m$ to $poly(l)$. 

A standard ``\textit{coupon collector}'' argument tells us how many samples we need to preserve the shortest $s\mhyphen t$ path.
For completeness we state and prove the precise version we need in Lemma~\ref{lem:couponcollector}.

\begin{lemma}
\label{lem:couponcollector}
    Given a random variable $E$ on a finite set $S$. Let $Y\subseteq S$ be such that for each $y\in Y$ $\mathbb{P}(E=y)\geq p$. Then the expected number \# of samples of E to get each element of $Y$ at least once is $\mathbb{E}(\#)=O(\frac{\log{|Y|}}{p})$.
\end{lemma}
\begin{proof}
    If $(i-1)$ elements of $Y$ have already been sampled then probability that a new element of $Y$ will be sampled is $p_i \geq (|Y|-(i-1))p$. $T[i]$ is the expected number of samples of $E$ before getting $i$ elements of $S$ after $i-1$ elements have already been sampled.
    \begin{equation}
        \mathbb{E}(T[i])=\frac{1}{p_i}\leq ((|Y|-(i-1))p)^{-1}
    \end{equation}
    It follows that to the expected number of samples before getting the whole set $Y$ is:
    \begin{equation}
        \mathbb{E}(T)=\sum_{i=1}^{|Y|}T[i]\leq \sum_{i=1}^{|Y|} \frac{1}{(|Y|-(i-1))p)}=\frac{1}{p}\sum_{k=1}^{|Y|}\frac{1}{k}=\Theta(\frac{\log{|Y|}}{p})
    \end{equation}
\end{proof}
\begin{algorithm}[h!]
\caption{Hybrid algorithm for finding the shortest $s\mhyphen t$ path in the graph $G$}\label{alg:divideandconquer}

\begin{algorithmic}[2]

\Statex\textbf{Input:} Unweighted graph $G$ with a path $P$ between vertices $s$ and $t$ satisfying \textbf{condition 1}, upper bound $l$ on the length of $P$.  
\Statex \textbf{Output:} the shortest $s\mhyphen t$ path $P$.
    \vspace{0.2cm}
    \begin{enumerate}
        
       \item Prepare $O(l\log{l})$ copies of the flow state $\ket{f^{s\rightarrow t}}$ to accuracy $\varepsilon=\frac{1}{\sqrt{8l}}$.
       \item Measure each of the states from 1. in the edges basis. Store the set of at most $\lceil 2l\log{l}\rceil$ distinct edges. 
        \item Run a classical path finding algorithm on the stored set of edges and output the result. 
    \end{enumerate}
   \end{algorithmic}
\end{algorithm}

    The effective resistance $R_G(s,t)$ is never larger than the (topological) length of the shortest $s\mhyphen t$ path. To convince oneself of that fact one can consider only the shortest $s$-$t$ path in G. The resistance of that path is equal exactly to its topological length (or weight in a weighted graph) and from \textit{Rayleigh monotonicity principle} one knows that addition of a new edge can decrease the resistance or leave it unchanged, but can never increase it. It follows that $R_G(s,t)$ always satisfies: \begin{equation}
    \label{eq:R_G<l}
        R_G(s,t)\leq l
    \end{equation}

\begin{proof}[Proof of Theorem~\ref{thm:simplealg}]
    
    For the considered graphs, from the above realisation (eq \eqref{eq:R_G<l}) and Lemma~\ref{lem:flowhalf} it follows directly that when measuring a perfect flow state, the probability of sampling an edge $e$ on the short path $P$ is 
    \[\frac{f_e^2}{R_G(s,t)}\geq \frac{1}{4l}.\]
    We require the accuracy of flow state preparation to be at least $\varepsilon=\frac{1}{\sqrt{8l}}$, so that the probability of sampling each edge on $P$ is still at least $\frac{1}{4l}-\frac{1}{8l} \geq \frac{1}{8l}$.
    Then from Lemma~\ref{lem:couponcollector} it follows that the expected number of samples to get the whole path $P$ is $\mathbb{E}(T)=O(l\log{l})$.  
    
    Let $S_\beta$ be the set of sampled edges after getting $\beta\mathbb{E}(T)$ samples for some positive constant $\beta$. Then via Markov's inequality we can see that after the expected number of steps the probability that shortest path is not in the set of sampled edges is bounded as $\mathbb{P}(P\not \subseteq S_{\beta})\geq \frac{1}{\beta}$.  Using a deterministic classical algorithm to find the short path in $S_{\beta}$, the overall algorithm therefore succeeds with constant probability. 
    
    By Theorem~\ref{thm:prepareflowstate}, each electric flow state preparation can be done in time $O(\sqrt{lm}(\frac{1}{\varepsilon}+\log{l d_s}) = O(l\sqrt{m})$. 
    This is repeated $O(l \log l)$ times, so the total algorithm takes $O(l^2\sqrt{m}\log l +\kappa)$ steps where $\kappa$ is the post-processing cost which depends on the method, but is negligible in comparison to $l^2\sqrt{m}\log l$ (for example it can be upper-bounded as $\kappa=O(l^2\log^2{l})$, in the case of Dijkstra).

\end{proof}

\section{Algorithm $\mathcal{A}_2$: divide and conquer}
\label{sec:algorithm2}

The algorithm we outline below relies on sampling the quantum flow state, performing effective-resistance-based validation checks on sampled edges and following a divide and conquer procedure. The algorithm comes in two flavours depending on the desired time-memory trade-off. Namely, whenever one has limited work-space access then the problem can be encoded on $\log{m}$ qubits and the algorithm will output the solution in $O(l\sqrt{m}\log^4{l})$ steps. However, if one is able to accommodate encoding on $O(l\log{m})$ qubits then the algorithm outputs the solution in $\tilde{O}(\sqrt{lm})$ steps. The latter query complexity is optimal in the assumed framework up to \textit{poly(log)} factors, as it represents the fundamental cost of a single quantum flow state preparation (Theorem~\ref{thm:prepareflowstate}), which is also equal to the cost of detecting the $s\mhyphen t$ path~\cite{belovs2013quantum}. The speedup over algorithm $\mathcal{A}_1$ comes from the fact that the cost of each successive preparation is smaller, as assured by the divide and conquer approach.
The algorithm's workflow consists of several steps. Firstly, having an upper bound on the length of the considered path one prepares $k=O(\log{l})$ quantum flow states $\ket{f^{s\rightarrow t}_G}$ to accuracy $\varepsilon=O(1/\sqrt{l})$ using the algorithm from Theorem~\ref{thm:prepareflowstate}, followed by computational basis measurements. As a next step of the algorithm for each of the sampled edges $e$ one creates a graph $G'= G\setminus e$ and runs the effective resistance estimation algorithm. Afterwards, one is in possession of $k$ values of which the maximum value must be found. Since the maximum finding and the comparison procedure can be done in logarithmic time (as the instance size is at most $O(\log{l})$) the complexity of that step will be omitted in further analysis. The edge $e'$ corresponding to the maximal value of effective resistance $R_{G\setminus e'}(s,t)$ is with high probability an edge on the shortest path, i.e $e'\subset P$ (as assured by corollary 1), and as we are about to show $e'$ is also most likely to be in a certain restricted distance from the middle of the path $P$. In step 3. one picks a random vertex from $e^{\prime}$ and repeats from the first step, without the effective resistance estimation $R_G(s,t)$ which does not need to be estimated again. After the expected number of iterations the algorithm terminates and with high probability outputs a positive witness to the shortest $s\mhyphen t$ path problem.

\begin{algorithm}
\caption{Quantum algorithm for finding the shortest $s\mhyphen t$ path in the graph $G$}\label{alg:divideandconquer}

\begin{algorithmic}[1]
\Statex\textbf{Input:} Graph $G$ with path $P$ connecting vertices $s$ and $t$ and satisfying condition~\ref{cond:1}, upper bound $l$ on the shortest path length. 
\Statex \textbf{Output:} the shortest $s\mhyphen t$ path $P$.
    \vspace{0.3cm}
    \begin{enumerate}
    \item Let $x=s$, $y=t$ and $\hat{R}_G(x,y)=l$ be an upper bound on $R_G(x,y)$.
        \item Prepare the flow state $\ket{f^{x\rightarrow y}}$ to accuracy $\varepsilon=\frac{1}{4k}$, measure in computational basis and store the resulting edge in a classical register. Repeat $k=O(\log{l})$ times. 
        \item For each of the edges sampled in step 2. remove the measured edge $e$ from $G$ and estimate the $x\mhyphen y$ effective resistance $R_{G\setminus e}(x,y)$ to multiplicative accuracy $\alpha$ and probability of failure $\delta_1=O(1/kl)$. Store the edge $e^{\star}$ corresponding to $\max_e(R_{G\setminus e}(x,y))$.
        \item If $e^{\star} \neq xy$, let $v$ be a random endpoint of $e^{\star}$ and partition the problem into two sub-problems, namely: paths from $s$ to $v$ and from $v$ to $t$. Calculate constant multiplicative accuracy approximations to $R_G(x,v)$ and $R_G(v,x)$ with probability of failure $\delta_2=O(1/l)$. Repeat from step $2$ on the created subproblems.

    \end{enumerate}
   \end{algorithmic}
\end{algorithm}

 \subsection{Algorithm $\mathcal{A}_2$: complexity and errors}
Next we demonstrate some conditions, which cause the divide and conquer algorithm (Algorithm~\ref{alg:divideandconquer}) to output the shortest $s\mhyphen t$ path, in the desired running time.

\begin{lemma}
    \label{lem:divideandconquer}
    Let $P$ be the shortest path between vertices $s$ and $t$. 
    Suppose there exist constants $\alpha, \beta, \gamma \in \Omega(1)$ such that for any $x,y \in P$, 
    \begin{enumerate}
        \item if an edge $e$ is sampled from the electric flow between $x$ and $y$,
    \[\Pr_{e\sim \text{elfs}(x,y)}[ R_{G\setminus e}(x,y) > (1+\alpha)R_P(x,y)] \ge \beta\]
    \item $R_G(x,y) / R_P(x,y) \ge \gamma $.
    \end{enumerate}
    Then Algorithm~\ref{alg:divideandconquer} finds the path $P$ in time $\tilde{O}(l\sqrt{m})$ using $O(\log{n})$ space. Furthermore the algorithm can be parallelised to succeed in time $\tilde{O}(\sqrt{lm})$ using $O(l\log{n})$ space.
\end{lemma}    
\begin{proof}

First we prove correctness of the algorithm, then we will prove the claimed run time.

\textbf{Correctness} When the algorithm succeeds, each iteration of steps 2.-4. outputs a new vertex $v$ on the shortest path $P$ between $x$ and $y$ and so there are at most $l$ iterations in total.
It therefore suffices to prove that each iteration of steps 2. and 3. returns an edge $e^{\star} \in P$ between $x$ and $y$ with probability at least $1-O(1/l)$; so that the overall algorithm returns the whole path $P$ with at least constant probability.

Let's analyse what happens in step 3. The edge $e^{\star}$ is chosen as the edge that maximises $R_{G\setminus e}(x,y)$ among the sampled edges.
For an edge $e\notin P$, by the Rayleigh monotonicity principle, $R_{G\setminus e}(x,y) \leq R_P(x,y)$. Therefore if an edge $e$ is sampled that satisfies $R_{G\setminus e}(x,y) > (1+\alpha)R_P(x,y)$, then it must be the case that $e$ is on the short path $P$.
Since the effective resistance calculations in step 3 are done to multiplicative accuracy $\alpha$, they will correctly distinguish between these two cases. So if there is at least one edge that satisfies $R_{G\setminus e}(x,y) > (1+\alpha)R_P(x,y)$ and all $k$ of the effective resistance calculations succeed, then $e^{\star}$ will be in $P$.
Each calculation fails with probability at most
$\delta_1 =\frac{1}{kl}$, so the probability that all $k=O(\log l)$ succeed is at least $1-1/l$.

    By the assumption 1 of the lemma,  the probability that an edge $e$, perfectly sampled from the electric flow state, satisfies $R_{G\setminus e}(x,y) > (1+\alpha)R_P(x,y)$ is at least $\beta$. 
    Since the electric flow state is prepared to accuracy $\epsilon= O(1/{\log{l}})$, each edge sampled in step 2 has this property with probability at least $\beta-O(1/\log{l}) \geq \beta/2$.
    After taking $k=O(\log{l})$ samples, the probability that no edges have this property is at most $ (1-\beta/2)^{O(\log{l})}\leq O(\frac{1}{l})$, and thus the probability that $e^{\star}$ is in $P$ is at least $1-O(\frac{1}{l})$.

\textbf{Running time} 
Next we use the second assumption of the lemma to show that $e^{\star}$ is likely to be not just from $P$, but from the ``middle section'' $\zeta$ of $P$ excluding the edges at distance $\leq \gamma R_P(x,y)/ 8 k$ away from $x$ and $y$. 
From eq \eqref{eq:electricflowstate}, the probability that an edge sampled from the electric flow state is in the ``bad region'' $P\setminus \zeta$ is 
\begin{equation}
\mathbb{P}[e\in P \setminus \zeta]
=\frac{\sum_{e\in P\setminus\zeta}f^2_er_e}{R_G(s,t)} 
\leq \frac{\sum_{e\in P \setminus \zeta} r_e}{R_G(x,y)}
=\frac{1}{R_G(x,y)}\cdot \frac{2\gamma R_P(x,y)}{8 \log{l}}
\end{equation}
where the inequality holds since the maximal flow through an edge is at most 1 in the unit electric flow, and the final equality comes from the definition of $\zeta$.
By assumption 2 of the lemma, $R_G / R_P \ge \gamma $ and so we have 
\[\Pr(e \in P\setminus \zeta) \leq \frac{1}{ 4 k}\]
Each edge sampled in step 2 of the algorithm is sampled not from the perfect electric flow state, but an $\epsilon=\frac{1}{4 k}$ approximation and so each of these edges is in $P\setminus \zeta$ with probability $1/2 k)$.

By a union bound, the probability that none of the $k$ edges sampled in step 2 are in $P\setminus \zeta$ is at most $1/2$, and so $\Pr(e^{\star} \in P\setminus \zeta) \leq 1/2$.

Therefore 
\[\Pr(e^{\star} \in \zeta) =\Pr(e^{\star} \in P) -\Pr(e^{\star} \in P\setminus \zeta) \geq 1/4\]
To analyse the query complexity, we first observe that step 2. includes $k$ state preparations each of which takes $O(\sqrt{\hat{R}_G(x,y)m}\log l)$ queries. Step 3 consists of $k$ effective resistance estimates which each take $O(\sqrt{\hat{R}_G(x,y)m}\log(1/\delta)\alpha^{-3/2})=O(\sqrt{\hat{R}_G(x,y)m}\log l)$ steps for constant $\alpha$. Step 4 involves two further effective resistance estimates and takes $O(\sqrt{\hat{R}_G(x,y)m}\log l)$ steps. Each iteration of steps 2-4 therefore takes a total of $O(\sqrt{\hat{R}_G(x,y)m}\log^2 l)$ steps.

If $e^{\star} \in \zeta$, then the path is subdivided into two sections of length at most $wR_P(x,y)$, where $w=(1-\frac{\gamma}{8k})$. As shown above, this happens with constant probability and so the expected number of iterations of steps 2-4. until all sections are of length $wR_P(x,y)$ is at most constant. The total cost of splitting into subproblems of length $wR_P(x,y)$ is therefore $O(\sqrt{\hat{R}_G(x,y)m}\log^2 l)$ steps. 

We can now calculate the cost of the whole algorithm.
While in practice the algorithm is run until all problems are subdivided until they reach a single edge, for the calculation of the total cost, we split into ``layers'', where in each layer the problems are subdivided until they are reduced in length by a multiplicative factor of $w$. 
The total number of layers is therefore $\log_{\frac{1}{w}}l = O(\log^2 l)$.

If the lengths of the subproblems at the start of a layer are $l_i$ for $i=1,2,\dots$ then the total cost of that layer is 
\[\sum_i O(\sqrt{l_i m}\log^2 l) \leq \sum_i O(l_i\sqrt{ m}\log^2 l)=O(l\sqrt{ m}\log^2 l)\]
Multiplying by the total number of layers, $O(\log^2 l)$, gives the total claimed run time of $O(l\sqrt{ m}\log^4 l)$.

\end{proof}

In particular, note that the cost of the last iteration dominates the total complexity of algorithm $\mathcal{A}_2$.\\
If one had a larger work space than $O(\log{n})$ then simply by parallelizing all procedures yields an algorithm that requires only $O(\sqrt{lm})$ steps but utilises $O(l\log{n})$ space. The space requirements are exponentially larger whenever $l>O(n^c)$ for any $c>0$. But for instances where $l=O(\log{n})$ the algorithm outputs the shortest path in $\tilde{O}(\sqrt{lm})$ steps and uses only $O(\log^2{n})$ space. This is an important observation, as it demonstrates that there exist graph instances, where one can find a path in the number of steps asymptotically equal (again, up to \textit{polylog} factors) to the cost of the quantum walk algorithm of Belovs~\cite{belovs2013quantum} for detecting a path, moreover one can do that while preserving the (efficient) \textit{polylog} work space.

\subsection{Algorithm 2 is efficient under condition 1$'$}
For Algorithm~\ref{alg:divideandconquer}, we require a slight strengthening of Condition~\ref{cond:1}, which we call Condition 1$'$.

\begin{taggedcondition}{1$'$}
    
    \label{cond:1'}
    For any subgraph $X \subset G$ such that $P \not\subseteq X$, we have $R_X(s,t) > (1+\alpha)R_P(s,t)$ for some constant $\alpha>0$.
\end{taggedcondition}

Note that for any fixed graph $G$ and path $P$ satisfying Condition~\ref{cond:1}, there exists some $\alpha>0$ such that Condition~\ref{cond:1'} is satisfied. However, as $G$ grows, it is possible that $\alpha \rightarrow 0$, and this will affect the performance of algorithm $\mathcal{A}_2$. For convenience, we therefore restrict to families of graphs where Condition~\ref{cond:1'} is satisfied for a constant $\alpha>0$, independent of $G$ and $P$.

We observe that if Condition~\ref{cond:1'} holds on the path $P$ between $s$ and $t$, then it also holds between any subpath of $P$ between vertices $x$ and $y$ on $P$. Indeed, let $e$ be an edge on $P$ between $x$ and $y$. Note that 
\begin{align} 
R_P(s,t) &=R_P(s,x)+R_P(x,y)+ R_P(y,t) \\
R_{G\setminus e}(s,t) &\leq R_P(s,x)+R_{G\setminus e}(x,y)+ R_P(y,t)
\end{align}
and so $R_{G\setminus e}(x,y) - R_P(x,y)\geq R_{G\setminus e}(s,t) - R_P(s,t) \geq \alpha R_P(s,t) \geq \alpha R_P(x,y)$.

To prove that Algorithm~\ref{alg:divideandconquer} is succeeds and is efficient under condition~\ref{cond:1'} , we need to show that the requirements of Lemma~\ref{lem:divideandconquer} are met.
First note that  Lemma~\ref{lem:RG>RP/2} implies that $R_G>R_P/2$ i.e.  requirement 2 of Lemma~\ref{lem:divideandconquer} is met  is met with $\gamma=1/2$.
By Condition~\ref{cond:1'}, every edge $e \in P$ satisfies $R_{G\setminus e} \> (1+\alpha)R_P$, and so 
\[\Pr_{e\sim \elfs{s}{t}}[R_{G\setminus e} \> (1+\alpha)R_P] \geq \Pr_{e\sim \elfs{s}{t}}[e \in P] \geq \frac{R_G(s,t)}{R_P(s,t)} \geq \frac{1}{2}  \]
where the second inequality is Lemma~\ref{lem:probinsubgraph} and the final inequality is Lemma~\ref{lem:RG>RP/2}.

\subsection{Algorithm 2 is efficient under condition 2}

In this section we prove that if a classical random walk is likely to find the short path, then the divide and conquer algorithm (Algorithm~\ref{alg:divideandconquer}) will work efficiently. More precisely, we show that the requirements of Lemma~\ref{lem:divideandconquer} are met when Condition~\ref{cond:2} holds.

First we observe that requirement 2 of Lemma~\ref{lem:divideandconquer} holds due to Lemma~\ref{lem:probonpath}. It therefore remains to show that requirement 1 of Lemma~\ref{lem:divideandconquer} also holds: we do this in Lemma~\ref{lem:cond2alg2}. 
But first we need one more technical result to understand how much the effective resistance changes when an edge $e$ is removed.
\begin{lemma}
    \label{lem:RGnoe}
    Let $f$ be the unit electric flow from $s$ to $t$ in $G$, with effective resistance $R_G$. 
    Then the effective resistance between $s$ and $t$ in $G\setminus e$ is 
    \[R_{G\setminus e} = R_G +\frac{f_e^2r_e}{1-R_e/r_e}\]
    where $R_e$ is the effective resistance of the unit electric flow between the endpoints of $e$.
\end{lemma}

\begin{proof}
    Let $f^{G(a\rightarrow b)}$ be the unit electric flow from $a$ to $b$, where $a$ and $b$ are the endpoints of the edge $e$. So $R_e$ is the effective resistance/energy of this flow. The amount of flow down the edge $e$ is $f^{G(a \rightarrow b)}_e=R_e/r_e$, and therefore $\frac{1}{1-R_e/r_e}f^{G( a \rightarrow b)}$ sends a unit flow in $G\setminus e$ from $a$ to $b$.

    The unit electric flow $f^{G\setminus e(s\rightarrow t)}$ from $s$ to $t$ in $G\setminus e$ is the linear combination:
    \[f^{G\setminus e(s\rightarrow t)} = f + \frac{f_e}{1-R_e/r_e}f^{G(a \rightarrow b)}\]

    Therefore the resistance $R_{G\setminus e}$ is 
    \[R_{G\setminus e}= (v^{G(s\rightarrow t)}_s-v^{G(s\rightarrow t)}_t)+ \frac{f_e}{1-R_e/r_e}(v^{G(a\rightarrow b)}_s-v^{G(a\rightarrow b)}_t)\]

    It remains to note that $(v^{G(s\rightarrow t)}_s-v^{G(s\rightarrow t)}_t)=R_G$ and 
    \[v^{G(a\rightarrow b)}_s-v^{G(a\rightarrow b)}_t=(\bra{a}-\bra{b})L^{-1}(\ket{s}-\ket{t})=v^{G(s\rightarrow t)}_a-v^{G(s\rightarrow t)}_b=f_er_e\]

\end{proof}

 We can now prove Lemma~\ref{lem:cond2alg2}:

\begin{lemma}
    \label{lem:cond2alg2}
    Suppose Condition~\ref{cond:2} holds for the short path $P$ between $s$ and $t$.
    Let $e$ be sampled from the electric flow state between $s$ and $t$. Then there exists $\alpha, \beta =\Omega(1)$ such that 
    \[\Pr_{e\sim \text{elfs}(s,t)}[ R_{G\setminus e}(s,t) > (1+\alpha)R_P(s,t)] \ge \beta\]
\end{lemma}

\begin{proof}
    In this proof, we drop the $(s,t)$ labels for conciseness.
    Let $X(e)=\frac{R_G}{R_{G\setminus e}-R_G}$, and $x^*=\frac{R_G}{(1+\alpha)R_P-R_G}$, so that $R_{G\setminus e} > (1+\alpha)R_P$ is equivalent to $X(e) < x^*$. 
    Then the probability we need to lower bound is $\Pr[X(e) < x^*]$ where the probability distribution is over the edge $e$ sampled from the electric flow.
    \[\Pr[X < x^*]\ge\Pr[X < x^* \text{ and } e \in P]=\Pr[X < x^*|e \in P]\Pr[e\in P]\]
    In fact, the above inequality is an equality because if $e\notin P$, then the path $P$ still exists when $e$ is removed from $G$ and so in this case $R_{G\setminus e} \le R_P$.
    Since $\Pr[e\in P]\ge 0.537$ by Lemma~\ref{lem:probonpath} and the assumption that Condition~\ref{cond:2} holds, it suffices to prove that $\Pr[X< x^*|e \in P]=\Omega(1)$ for some choice of $\alpha$.
    We do this with Markov's inequality:
    \[\Pr[X< x^*|e \in P]=1-\Pr[X \ge x^*|e \in P] \ge 1- \frac{\E[X |e\in P]}{x^*} \]
    Recall that the probability of selecting edge $e$ is $p(e)=\frac{f_e^2r_e}{R_G}$, and by Lemma~\ref{lem:RGnoe} $X(e)=(1-R_e/r_e)/p(e)$, so 
    \[\E[X |e\in P]=\frac{1}{\Pr[e\in P]}\sum_{e \in P }p(e)X(e) =\frac{1}{\Pr[e\in P]}\sum_{e \in P }(1-R_e/r_e)\]
    Now we bound $\sum_{e \in P }(1-R_e/r_e)$ using the inequality $1-y \le \ln(1/y)$ 
    \[\sum_{e \in P }(1-R_e/r_e)\le \sum_{e \in P }\ln(r_e/R_e)=\ln \left(\prod_{e \in P} r_e/R_e\right)\]

    Recall that for $T$ a random spanning tree, $\Pr_{T}[e\in T]=R_e/r_e$ (equation \eqref{eqn:PeinT}), and so by the negative association property (equation \eqref{eqn:negativeassociation}) $\prod_{e \in P} R_e/r_e \ge \Pr_T[P\subseteq T]$.
    Therefore

    \[\E[X |e\in P]=\frac{1}{\Pr[e\in P]}\sum_{e \in P }(1-R_e/r_e)\le \frac{\ln(1/\Pr_T[P\subseteq T])}{ \Pr[e\in P]}\]
    Let $q=\Pr_T[P\subseteq T]$ and use Lemma~\ref{lem:probonpath}: $\Pr[e\in P]\ge R_G/R_P \ge q$, to show that $1/x^*=(1+\alpha)R_P/R_G-1\le(1+\alpha)\frac{1}{q}-1$ and so 
    \[\frac{\E[X |e\in P]}{x^*}\le ((1+\alpha)\tfrac{1}{q}-1)\frac{\ln(1/q)}{q}
    \]
    This is less than 1 when
    \[(1+\alpha)\tfrac{1}{q}-1\le \frac{q}{\ln(1/q)}\]
    \[\Leftrightarrow\alpha\le \frac{q^2}{\ln(1/q)} +q-1\]
    It is easy to check that $\frac{q^2}{\ln(1/q)} +q-1$ is an increasing function of $q$. Numerically we can compute that the only zero of this function is at $\approx 0.5368\dots$.
    By the assumption that Condition~\ref{cond:2} is satisfied, we have $q\ge 0.537$ and so we can choose $\alpha$ such that 
    \[0< \alpha < \frac{q^2}{\ln(1/q)} +q-1\].
\end{proof}

\section{Remarks \& Open Problems}
\label{sec:conclusion}
We have introduced quantum algorithms that achieve asymptotically improved runtime for the shortest \( s\text{-}t \) path problem on structured graph instances, surpassing all previously known classical and quantum methods. Notably, one of our algorithms finds the shortest path with asymptotic runtime matching that of path detection algorithms, up to \textit{polylogarithmic} factors. Furthermore, the electric network framework at the core of our approach has been shown to be effective for the minimum cut problem as well~\cite{path-edge}, indicating its broader applicability. We anticipate that this framework can be further utilised to obtain additional quantum complexity advantages for related graph problems.\\

The following open questions remain:
\begin{enumerate}

    \item Understanding what classes of graphs satisfy \textbf{condition \ref{cond:1}} and \textbf{condition \ref{cond:2}}. In particular, do there exist graphs where $m=O(n^{1+\varepsilon})$ where \textbf{condition \ref{cond:1}} is satisfied between every pair of vertices?

    \item  Moreover, we would like to understand whether there exist path finding algorithms that achieve the asymptotic runtime of path detection algorithms for unrestricted (or less restricted) instances of the problem.

\end{enumerate}

\section{Acknowledgements}
We would like to thank Simon Apers for the insightful discussion during his visit to Royal Holloway UoL.

\printbibliography

\appendix
\section{Consequences of Condition 1}
\label{sec:appendix}
In this appendix we provide the missing proofs for Lemma~\ref{lem:flowhalf} and Lemma~\ref{lem:RG>RP/2}.
 \begin{proof}[Proof of Lemma~\ref{lem:flowhalf}]
Consider the unit electric $s\mhyphen t$ flow $f$ and fix an edge $e \in P$. 
We can decompose $f$ as a linear combination of two unit flows $f^P$ and $f^C$ as 
\[f=f_ef^P +(1-f_e) f^C\] where $f^P$ is the unit $s\mhyphen t$ flow along $P$ and $f^C$ is defined as
\[f^C=\frac{f- f_ef^{P}}{(1-f_e)}.\]

Note that $f^C$ has no flow along edge $e$ ($f^C_e=0$) and hence can be thought of as an $s\mhyphen t$ flow in the graph $G\setminus e$.


Now we consider the $s\mhyphen t$ flow $f(\alpha)$ which is parameterised by $\alpha \in [-1,1]$:

\begin{equation}
    f(\alpha)=\alpha f^{P}+(1-\alpha)f^{C}.
    \label{eq:falpha}
\end{equation}
At $\alpha=f_e$, this is the electric flow. Since the electric flow minimises the energy among \emph{all} unit $s\mhyphen t$ flows, we can compute $f_e$ by finding the value of $\alpha$ that minimises the energy $\mathcal{E}(f(\alpha)) = f(\alpha)\cdot f(\alpha)$ eq \eqref{eq:energydotproduct}.

Substituting in eq \eqref{eq:falpha} gives
\begin{equation}
    \mathcal{E}(f(\alpha)))= f(\alpha)\cdot f(\alpha)
    =\alpha^2(f^{P}\cdot f^{P})+2\alpha(1-\alpha) f^{P}\cdot f^{C}+(1-\alpha)^2 f^{C}\cdot f^{C})
\end{equation}
Since this is quadratic in $\alpha$, the minimum can easily be found by completing the square or solving $\frac{d}{d\alpha}\mathcal{E}(f(\alpha))=0$.
This leads to 
\begin{equation}
    f_e = \arg \min_{\alpha} \mathcal{E}(f(\alpha)) = \frac{f^C \cdot f^C - f^P\cdot f^C}{f^P\cdot f^P + f^C \cdot f^C -2 f^C \cdot f^P }
    \label{eq:feargmin}
\end{equation}

The flow $f^P$ is the electric $s\mhyphen t$ flow in $P$, so $f^P \cdot f^P=\mathcal{E}(f^P)=R_P$; and $f^C$ is a flow on $G\setminus e$, so $f^C\cdot f^C=\mathcal{E}(f^C)\ge R_{G\setminus e}$.
Combining this with Condition \ref{cond:1} gives 
\begin{equation}
    f^{P}\cdot f^{P}=R_P\le R_{G\setminus e} \le f^{C}\cdot f^{C}
    \label{eq:fPfCinequality}
\end{equation}

Therefore the denominator of \eqref{eq:feargmin} satisfies $ f^P\cdot f^P + f^C \cdot f^C -2 f^C \cdot f^P \le 2(f^C\cdot f^C -f^C \cdot f^P)$. Noting that the denominator is also $(f^P-f^C)\cdot (f^P-f^C) \ge 0$, we therefore have the claimed result $f_e \ge 1/2$.

\end{proof}

Finally we prove Lemma~\ref{lem:RG>RP/2}, which follows as a simple consequence of the following:

\begin{lemma}
    \label{lem:1/R_G}
    Let $P$ be a path in $G$ such that Condition~\ref{cond:1} holds. Then there exists a probability distribution $\{p(e)\}_{e \in P}$ over the edges of $P$ such that:
    \begin{equation}
        \frac{1}{R_G} \leq \frac{1}{\mathbb{E}_{e \sim p(e)}(R_{G\setminus e})} +\frac{1}{R_P} 
        \label{eq:1/R_G}
    \end{equation}
\end{lemma}

Before we prove Lemma~\ref{lem:1/R_G}, let's first observe how it implies Lemma~\ref{lem:RG>RP/2}. By Condition~\ref{cond:1}, $R_{G\setminus e} >R_P$ for all edges $e \in P$, and so $\mathbb{E}_{e \sim p(e)}(R_{G\setminus e}) >R_P$. Therefore the right hand side of \eqref{eq:1/R_G} is at most $2/R_P$, and the inequality of Lemma~\ref{lem:RG>RP/2} follows.
\begin{proof}[Proof of Lemma~\ref{lem:1/R_G}]
    Let $v^{G \setminus e_i}_a$ be the voltage at $a$ of the unit electric flow from $s$ to $t$ in $G \setminus e_i$. Taking the same voltages on the whole of $G$ induces an electric flow which is a unit flow from $s$ to $t$ and an additional $q_i$ flow from $x_{i-1}$ to $x_i$ along the edge $e_i$. 

    For now, assume that $q_i >0$, so that we can take the linear combination of voltages $\sum_{i=1}^l q_i^{-1} v^{G\setminus e_i}$.
    These voltages therefore induce an electric flow from $s$ to $t$ of total flow $1+\sum_{i=1}^l q_i^{-1}$ with potential difference between $s$ and $t$ of $\sum_{i=1}^l q_i^{-1} R_{G\setminus e_i}$. 
    Rescaling, we have that the potential difference of the unit electric flow (and hence the effective resistance) between $s$ and $t$ is 
    \[R_G=\frac{\sum_{i=1}^l q_i^{-1} R_{G\setminus e_i}}{1+\sum_{i=1}^l q_i^{-1}}.\]

    Let $p(e_i)= q_i^{-1}/(\sum_j q_j^{-1})$ and rearrange to get:

    \[\frac{1}{R_G} =\frac{1+\sum_{i=1}^l q_i^{-1}}{\sum_{i=1}^l q_i^{-1} R_{G\setminus e_i}}
    =\frac{1}{\sum_{i=1}^l q_i^{-1} R_{G\setminus e_i}}+ \frac{1}{\sum_{i=1}^l p(e_i) R_{G\setminus e_i}}\]
    
    The second term can be rewritten as $\sum_{i=1}^l p(e_i) R_{G\setminus e_i} = \mathbb{E}(R_{G \setminus e})$. For the first term, we note that by Ohm's law, the flow $q_i$ along edge $e_i$ is 
    \[q_i = \frac{v^{G\setminus e_i}_{x_{i-1}}-v^{G\setminus e_i}_{x_{i}}}{r_{e_i}}.\]
    Since $0 \leq v_x^{G\setminus e_i}\leq R_{G\setminus e_i}$, it follows that $q_i \leq R_{G\setminus e_i}/r_{e_i} $.
    Therefore $\sum_{i=1}^l q_i^{-1} R_{G\setminus e_i}\geq \sum_{i=1}^l r_{e_i} = R_P$.

    It remains to show that $q_i > 0$ (in fact we show that $q_i \geq 1$). Since the flow through edge $e_j$ is at most $1$, Ohm's law implies that 
    \[v^{G\setminus e_i}_{x_{j-1}}-v^{G\setminus e_i}_{x_{j}} \le r_{e_j}\]
    Taking the sum over $j\in \{1, \dots i-1\}$ and $j\in \{i+1, \dots l\}$, most terms cancel and (recalling that $v^{G\setminus e_i}_{x_0}=R_{G\setminus e_i}$ and $v^{G\setminus e_i}_{x_l}=0$):
    \[R_{G\setminus e_i} - v^{G\setminus e_i}_{x_{i-1}}  \leq \sum_{j=1}^{i-1} r_{e_j} \quad \text{ and } \quad v^{G\setminus e_i}_{x_{i}} - 0  \leq \sum_{j=i+1}^{l} r_{e_j}.\]
    Combining these two and rearranging we have:
    \[v^{G\setminus e_i}_{x_{i-1}} - v^{G\setminus e_i}_{x_{i}} \geq R_{G\setminus e_i} -  \left(\sum_{j=1}^{i-1} r_{e_j}  +\sum_{j=i+1}^{l} r_{e_j}\right) = R_{G\setminus {e_i}}-R_P +r_{e_i} \geq r_{e_i}\]
    where the final inequality follows from Condition~\ref{cond:1}.
    
\end{proof}

\end{document}